\documentclass[manyauthors]{fundam}

\publyear{26}
\papernumber{2}
\volume{196}
\issue{1}
\theDOI{10.46298/fi.11468}

\versionForARXIV

\usepackage{url} 
\usepackage[ruled,lined]{algorithm2e}
\usepackage{graphicx}
\usepackage{tikz}
\usetikzlibrary{automata, positioning, arrows}

\usepackage[english]{babel}
\usepackage[ansinew]{inputenc}
\usepackage{amssymb,amsmath}
\usepackage{psfrag}
\usepackage[all]{xy}
\usepackage{subfigure}
\usepackage[dvips]{epsfig}
\usepackage{caption}
\usepackage{enumerate}
\usepackage{xcolor}
\usetikzlibrary{matrix}
\usetikzlibrary{positioning}
\usetikzlibrary{decorations.markings}
\tikzset{main node/.style={circle,fill=blue!20,draw,minimum size=0.2cm,inner sep=0pt},}
\usepackage{centernot}
\usepackage{enumitem}
\usepackage{float}

\newcommand{\Z}{\mathbb{Z}}

\newcommand{\ID}{\gamma^{ID}}
\newcommand{\LD}{\gamma^{LD}}

\newcommand{\SLD}{\gamma^{SLD}}
\newcommand{\GSLD}{\gamma^{DLD}}

\begin{document}

\tikzset{square matrix/.style={
    matrix of nodes,
    column sep=-\pgflinewidth, row sep=-\pgflinewidth,
    nodes={draw,
      minimum height=#1,
      anchor=center,
      text width=#1,
      align=center,
      inner sep=0pt
    },
  },
  square matrix/.default=0.5cm
}

\tikzset{big square matrix/.style={
    matrix of nodes,
    column sep=-\pgflinewidth, row sep=-\pgflinewidth,
    nodes={draw,
      minimum height=#1,
      anchor=center,
      text width=28,
      align=center,
      inner sep=0pt
    },
  },
  big square matrix/.default=0.88cm
}

\title{New Optimal Results on Codes for Location in Graphs\thanks{An extended abstract~\cite{JLLnrclg} of the paper has been presented at the Fifth Russian Finnish Symposium on Discrete Mathematics.}\thanks{Research supported partially by the Research Council of Finland grants 338797 and 358718.}}


\author{Ville Junnila\\
Department of Mathematics and Statistics \\
University of Turku \\ Turku
FI-20014, Finland \\
viljun@utu.fi \\
\and Tero Laihonen\\
Department of Mathematics and Statistics \\
University of Turku \\ Turku
FI-20014, Finland \\
terolai@utu.fi
 \\
\and Tuomo Lehtil{\"a}\thanks{Research supported by the University of Turku Graduate School (UTUGS), the Vilho, Yrj{\"o} and Kalle V{\"a}is{\"a}l{\"a} Foundation, and the Jenny and Antti Wihuri Foundation.}\corresponding \\
Department of Mathematics and Statistics \\
University of Turku \\ Turku
FI-20014, Finland \\
tualeh@utu.fi
}


\maketitle

\runninghead{V. Junnila, T. Laihonen, T. Lehtil{\"a}}{New Optimal Results on Codes for Location in Graphs}

\begin{abstract}
In this paper, we broaden the understanding of the recently introduced concepts of solid-locating-dominating and self-locating-dominating codes in various graphs. In particular, we present the optimal, i.e., smallest possible, codes in the infinite triangular and king grids.  
Furthermore, we give optimal locating-dominating, self-locating-dominating and solid-locating-dominating codes in the direct product $K_n\times K_m$ of complete graphs. We also present optimal solid-locating-dominating codes for the Hamming graphs $K_q\square K_q\square K_q$ with $q\geq2$.
\end{abstract}


\section{Introduction}

Sensor networks consist of sensors monitoring various places and connections between these places (see~\cite{lowww}). A sensor network is modeled as a simple and undirected graph $G=(V(G),E(G))=(V,E)$. In this context, a sensor can be placed on a vertex $v$ and its closed neighbourhood $N[v]$ represents the set of locations that the sensor monitors. Besides assuming that graphs are simple and undirected, we also assume that they are connected and have cardinality at least two. In the following, we present some terminology and notation. The \textit{closed neighbourhood} of $v$ is defined as $N[v]=N(v)\cup \{v\}$, where $N(v)$ is the \textit{open neighbourhood} of $v$, that is, the set of vertices adjacent to $v$. A \textit{code} $C$ is a nonempty subset of $V$ and its elements are \textit{codewords}. The codeword $c\in C$ \emph{covers} a vertex $v\in V$ if $v\in N[c]$. We denote the set of codewords covering $v$ in $G$ by 
$$I(G,C;v) = I(C;v)=I(v)=N[v]\cap C \text{.}$$ 
The set $I(v)$ is called an \textit{identifying set} or an $I$\emph{-set}. We say that a code $C\subseteq V$ is \textit{dominating} in $G$ if $I(C;u)\neq \emptyset$ for all $u\in V$. If the sensors are placed at the locations corresponding to the codewords, then each vertex is monitored by the sensors located in $I(v)$. More explanation regarding location detection in the sensor networks can be found in~\cite{Trachtenberg,LT:disj,Ray}.

Let us now define \textit{identifying codes}, which were first introduced by Karpovsky \textit{et al.} in~\cite{kcl}. For numerous papers regarding identifying codes and related topics, the interested reader is referred to the online bibliography~\cite{lowww}.
\begin{definition}
A code $C \subseteq V$ is \emph{identifying} in $G$ if  we have $I(C;u) \neq \emptyset$ for all $u\in V$ and
\[
I(C;u) \neq I(C;v)
\]
for all distinct $u, v \in V$.
An identifying code $C$ in a finite graph $G$ with the smallest cardinality is called \emph{optimal} and the number of codewords in an optimal identifying code is denoted by $\ID(G)$.
\end{definition}

Identifying codes require unique $I$-sets for codewords as well as for non-codewords. However, if we omit the requirement of unique $I$-sets for codewords, then we obtain the following definition of \emph{locating-dominating codes}, which was first introduced by Slater in~\cite{RS:LDnumber,S:DomLocAcyclic,S:DomandRef}.
\begin{definition}
A code $C \subseteq V$ is \emph{locating-dominating} in $G$ if we have $I(C;u) \neq \emptyset$ for all $u\in V \setminus C$ and
\[
I(C;u) \neq I(C;v) 
\]
for all distinct $u, v \in V \setminus C$.
\end{definition}

Notice that an identifying code in $G$ is also locating-dominating (by the definitions). In \cite{JLLrntcld}, self-locating-dominating and solid-locating-dominating codes have been introduced and, in \cite{CubeCon,SLDDLDgraafit}, they have been further studied. The definitions of these codes are given as follows.

\eject

\begin{definition} Let $C \subseteq V$ be a code in $G$.
\begin{itemize}
\item[(i)] We say that $C$ is a \emph{self-locating-dominating code} in $G$ if for all $u \in V \setminus C$ we have $I(C;u) \neq \emptyset$ and
    \[
    \bigcap_{c \in I(C;u)} N[c] = \{u\} \text{.}
    \]
    \vspace*{-2.5ex}
\item[(ii)] We say that $C$ is a \emph{solid-locating-dominating code} in $G$ if for all distinct $u, v \in V \setminus C$ we have
    \[
    I(C;u) \setminus I(C;v) \neq \emptyset \text{.}
    \]
        \vspace*{-4ex}
\end{itemize}
\end{definition}

Observe that since $G$ is a connected graph on at least two vertices, a self-locating-dominating and solid-locating-dominating code is always dominating. Analogously to identifying codes, in a finite graph $G$, we say that dominating, locating-dominating, self-locating-dominating and solid-locating-dominating codes with the smallest cardinalities are \emph{optimal} and we denote the cardinality of an optimal code by $\gamma(G)$, $\LD(G), \SLD(G)$ and $\GSLD(G)$, respectively.

In the following theorem, we offer characterizations of self-locating-dominating and solid-locat\-ing-dominating codes for easier comparison of them.
\begin{theorem}[\cite{JLLrntcld}] \label{ThmCharacterizationSLDILD}
Let $G=(V,E)$ be a connected graph on at least two vertices:
\begin{itemize}
\item[(i)] A code $C \subseteq V$ is self-locating-dominating if and only if for all distinct $u \in V \setminus C$ and $v \in V$ we have
    \[
    I(C;u) \setminus I(C;v) \neq \emptyset \text{.}
    \]
\item[(ii)] A code $C \subseteq V$ is solid-locating-dominating if and only if for all $u \in V \setminus C$ we have $I(C;u)\neq \emptyset$ and
    \[
    \left( \bigcap_{c \in I(C;u)} N[c] \right) \setminus C = \{u\} \text{.}
    \]
\end{itemize}
\end{theorem}

Based on the previous theorem, we obtain the following corollary.
\begin{corollary} \label{CorollarySLdtoDLD}
If $C$ is a self-locating-dominating or solid-locating-dominating code in $G$, then $C$ is also solid-locating-dominating or locating-dominating in $G$, respectively. Furthermore, for a finite graph $G$, we have $$\LD(G)\leq \GSLD(G) \leq \SLD(G) \text{.}$$ 
\end{corollary}

The structure of the paper is described as follows. First, in Section \ref{Grids}, we obtain optimal self-locating-dominating and solid-locating-dominating codes in the infinite triangular and king grids, i.e., the smallest possible codes regarding their density (a concept defined later). Regarding the triangular grid, the proofs are rather simple and straightforward, but they serve as nice introductory examples to the concepts of solid-location-domination and self-location-domination. However, the case with the king grid is more interesting; in particular, the proof of the lower bound for solid-location-domination is based on global arguments instead of only local ones, which are more usual in domination type problems. Then, in Section~\ref{secDirProd}, we give optimal locating-dominating, self-locating-dominating and solid-locating-dominating codes in the direct product $K_n\times K_m$ of complete graphs, where $2\leq n\leq m$. Finally, in Section \ref{Hamming}, we present optimal solid-locating-dominating codes for graphs $K_q\square K_q\square K_q$ with $q\geq2$.

\section{Triangular and king grids}\label{Grids}

In this section, we consider solid-location-domination and self-location-domination in the so called infinite triangular and king grids. As defined in the introduction, for finite graphs, the optimality of a code has been defined using the minimum cardinality. However, this method is not valid for the infinite graphs of this section. Hence, we need to use the usual concept of \textit{density} of a code, for a recent thorough investigation on the concept of density in the infinite graphs, see~\cite{sampaio2024density}. Both triangular and king grids have been widely studied in the field of location and domination; for an extensive coverage on the topics, an interested reader is referred to the online bibliography~\cite{lowww}. In~\cite{honkala2006optimal}, Honkala has shown that the optimal density for locating-dominating codes in the infinite triangular grid is~$\frac{13}{57}$. In~\cite{honkala2006locating}, Honkala and Laihonen have shown that in the infinite king grid, the optimal density of locating-dominating codes is $1/5$. Furthermore, identifying codes have been considered in the king grid in~\cite{charon2002identifying, cohen2001codes} and in the triangular grid in~\cite{kcl}. The optimal densities are $ 2/9$ and $1/4$, respectively.

Let us first consider the \textit{infinite triangular grid}.
\begin{definition}
Let $G=(V,E)$ be a graph with the vertex set
\[
V=\left\{i(1,0)+j\left(\frac{1}{2},\frac{\sqrt{3}}{2}\right)\mid i,j\in \Z\right\}
\]
and two vertices are defined to be adjacent if their Euclidean distance is equal to one. The obtained graph $G$ is called the \textit{infinite triangular grid} and it is illustrated in Figure~\ref{Triangular grid}. We further denote $v(i,j)= i(1,0) +j \left(\frac{1}{2}, \frac{\sqrt{3}}{2}\right)$. Let $R_n$ be the subgraph of $G$ induced by the vertex set $V_n = \{v(i,j) \mid |i|, |j| \leq n \}$. The \emph{density} of a code in $G$ is now defined as follows:
\[
D(C) = \limsup_{n \rightarrow \infty} \frac{|C\cap V_n|}{|V_n|}.
\]
We say that a code is \emph{optimal} if there exists no other code with smaller density.
\end{definition}

%

\begin{figure}[H]
  \centering
    \begin{tikzpicture}
\clip (-2.7,-0.7) rectangle (10.2,2.45);

\node[main node, minimum size=0.3cm] (1)[fill=white] {};

\node[main node, minimum size=0.3cm] (2) [above right = 0.866cm and 0.5cm of 1]  [fill=white]{};
\node[main node, minimum size=0.3cm] (3) [above right = 0.866cm and 0.5cm of 2]  [fill=white]{};

\node[main node, minimum size=0.3cm] (4) [right = 1cm  of 1][fill=white] {};
\node[main node, minimum size=0.3cm] (5) [right = 1cm  of 4][fill=white] {$u$};
\node[main node, minimum size=0.3cm] (6) [right = 1cm  of 5][fill=white] {};
\node[main node, minimum size=0.3cm] (7) [right = 1cm  of 6][fill=white] {};
\node[main node, minimum size=0.3cm] (8) [right = 1cm  of 7][fill=white] {};

\node[main node, minimum size=0.3cm] (9) [right = 1cm  of 2][fill=white] {$v$};
\node[main node, minimum size=0.3cm] (10) [right = 1cm  of 9][fill=white] {};
\node[main node, minimum size=0.3cm] (11) [right = 1cm  of 10][fill=white] {};
\node[main node, minimum size=0.3cm] (12) [right = 1cm  of 11][fill=white] {};
\node[main node, minimum size=0.3cm] (17) [right = 1cm  of 12][fill=white] {};

\node[main node, minimum size=0.3cm] (13) [right = 1cm  of 3][fill=white] {};
\node[main node, minimum size=0.3cm] (14) [right = 1cm  of 13][fill=white] {$w$};
\node[main node, minimum size=0.3cm] (15) [right = 1cm  of 14][fill=white] {};
\node[main node, minimum size=0.3cm] (16) [right = 1cm  of 15][fill=white] {};


    \path[draw,thick]
    (1) edge  node {} (2)
      edge  node {} (4)
      (2)  edge  node {} (3)
      (2)  edge  node {} (4)
			edge  node {} (9)
 	(3)	edge  node {} (13)
 		edge  node {} (9)
 		(9) edge  node {} (13)
 		edge  node {} (10)
 		edge  node {} (4)
 		edge  node {} (5)
 		(5) edge  node {} (4)
 		edge  node {} (6)
 		edge  node {} (10)
 		(10) edge  node {} (11)
 		edge  node {} (6)
 		edge  node {} (13)
 		edge  node {} (14)
 		(14) edge  node {} (13)
 		edge  node {} (15)
 		edge  node {} (11)
 		(11)	edge  node {} (15)
 		edge  node {} (6)
 		edge  node {} (7)
 		edge  node {} (12)
 		(7)edge  node {} (12)
 		edge  node {} (6)
 		edge  node {} (8)
 		(12)edge  node {} (8)
 		edge  node {} (17)
 		edge  node {} (15)
 		edge  node {} (16)
 		(17)edge  node {} (8)
 		edge  node {} (16)
 		(15)edge  node {} (16)

;
    \end{tikzpicture}
\centering
\caption{Triangular grid with the vertices $v = v(0,0)$, $u = v(1,-1)$ and $w = v(1,1)$.}\label{Triangular grid}
\end{figure}

In the following theorem, optimal self-locating-dominating and solid-locating-dominating codes are given in the triangular grid. The methods used in the proof are rather typical for domination type of problems. However, we present the proof for completeness and as an introductory example.
\begin{theorem} \label{SLDkolmio}
Let $G=(V,E)$ be the triangular grid. The code
\[
C = \{v(i,j)\mid i, j \equiv 0 \pmod{2}\}
\]
is self-locating-dominating in $G$ and, therefore, also solid-locating-dominating. The density of the code $C$ is equal to $1/4$ and there exists no self-locating-dominating or solid-locating-dominating code with smaller density, i.e., the code is optimal in both cases.
\end{theorem}
%
\begin{proof}
Let us first show that the code $C$ is self-locating-dominating in the triangular grid $G$. The proof now divides into the following cases depending on the parity of $i$ and $j$ in $v(i,j)$:
\begin{itemize}
\item If $i$ is odd and $j$ is even, then $$I(v(i,j))=\{v(i-1,j),v(i+1,j)\} \text{ and } N[v(i-1,j)] \cap N[v(i+1,j)] = \{v(i,j)\}.$$
\item Analogously, if $i$ is even and $j$ is odd, then $$I(v(i,j))=\{v(i,j-1),v(i,j+1)\} \text{ and } N[v(i,j-1)] \cap N[v(i,j+1)] = \{v(i,j)\}.$$
\item Finally, if $i$ and $j$ are both odd, then $$I(v(i,j))=\{v(i-1,j+1),v(i+1,j-1)\} \text{ and } N[v(i-1,j+1)] \cap N[v(i+1,j-1)] = \{v(i,j)\}.$$
\end{itemize}
Thus, as $v(i,j)$ is a codeword for even $i$ and $j$, the code $C$ is self-locating-dominating in $G$. Furthermore, we have $D(C)=\frac{1}{4}$ since $v(i,j)$ is a codeword if and only if $i$ and $j$ are both even, showing that the density of self-locating-dominating codes in $G$ is at most $\frac{1}{4}$. Notice that $C$ is also a solid-locating-dominating code.


For the lower bound, assume that $C'$ is a solid-locating-dominating code in $G$. Immediately, by the definition of solid-locating-dominating codes, we know that $|I(C';u)| \geq 2$ for any non-codeword $u$ in the triangular grid. Therefore, by counting in two ways the pairs $(u,c)$, where $c \in C' \cap V_n$ and $u \in N[c] \cap V_{n-1}$, we first obtain that $7|C' \cap V_n| \geq |C' \cap V_{n-1}| + 2(|V_{n-1}| - |C' \cap V_{n-1}|)$. Indeed, since $|N[v]|=7$ for any vertex $v$, there are at most $7|C' \cap V_n|$ pairs $(u,c)$. Furthermore, since each codeword $c$ contributes one pair $(u,c)$ by choosing $u=c$ and since each non-codeword in $V_{n-1}$ contributes at least two pairs $(u,c)$, we obtain the inequality. We may modify the inequality further to
$|C' \cap V_{n-1}| + 2(|V_{n-1}| - |C' \cap V_{n-1}|)\geq 2|V_{n-1}| - |C' \cap V_{n-1}| \geq 2|V_{n-1}| - |C' \cap V_{n}|$, which implies $|C' \cap V_n| \geq |V_{n-1}|/4$. Thus, we may estimate the density of $C'$ as follows:
\[
D(C') = \limsup_{n \rightarrow \infty} \frac{|C'\cap V_n|}{|V_n|} \geq \limsup_{n \rightarrow \infty} \frac{|V_{n-1}|/4}{|V_n|} = \frac{1}{4} \text{.}
\]
\end{proof}

Next we consider the more interesting problems of solid-location-domination and self-location-domination in the infinite king grid. Let us first begin by defining the grid and the density of a code in it.
\begin{definition}
Let $G=(V,E)$ be a graph with $V=\Z^2$. For the vertices $v=(v_1,v_2)\in V$ and $u=(u_1,u_2)\in V$ with $u\neq v$, we have $vu\in E$ if and only if $|v_1-u_1|\leq1$ and $|v_2-u_2|\leq1$. The obtained graph $G$ is called the \textit{infinite king grid}. Further let $V_n$ be a subset of $V$ such that $V_n=\{(x,y)\mid |x|\leq n, |y|\leq n\}$. The \emph{density} of a code $C \subseteq V = \Z^2$ is now defined as$$D(C)=\underset{n\to \infty}{\lim\sup}\frac{|C\cap V_n|}{|V_n|}.$$ We say that a code is \emph{optimal} if there exists no other code with smaller density.
\end{definition}

\medskip

 In what follows, we first consider solid-location-domination in the king grid. In the following theorem, we present a solid-locating-dominating code in the king grid with density $1/3$. Later, in Theorem~\ref{ThmKingLowerBound}, it is shown that the code is optimal.
\begin{theorem} \label{ThmDLDKingConstruction}
Let $G=(V,E)$ be the king grid. The code
\[
C =  \left\{(x,y)\in \Z^2\mid |x|+|y|\equiv 0 \pmod 3\right\}
\]
is solid-locating-dominating in $G$ and its density is $1/3$.
\begin{proof}
Let $C=\left\{(x,y)\in \Z^2\mid |x|+|y|\equiv 0 \mod 3\right\}$ be a code in $G$ (illustrated in Figure~\ref{DLD in king grid}). By the definition, it is immediate that the density of $C$ is equal to $1/3$. In order to show that $C$ is a solid-locating-dominating code in $G$, we prove that the condition of Theorem~\ref{ThmCharacterizationSLDILD}(ii) holds for every non-codeword of $G$. Let $u = (x,y) \in \Z^2$ be a vertex not belonging to $C$. Suppose first that $x=0$ and $y > 0$. Now, if $y \equiv 1 \pmod{3}$, then 
\[
I(u) = \{u + (0,-1), u + (-1,1), u + (1,1)\}
\]
and 
\[
N[u + (0,-1)] \cap N[u + (-1,1)] \cap N[u + (1,1)] = \{u\},
\] else $y \equiv 2 \pmod{3}$ implying 
\[
I(u) = \{u + (-1,0), u + (1,0), u + (0,1)\}
\] and 
\[
(N[u + (-1,0)] \cap N[u + (1,0)] \cap N[u + (0,1)]) \setminus C = \{u\}.
\] 
Thus, the required condition is met. The case with $y < 0$ is analogous. Moreover, the case with $y=0$ is symmetrical to the one with $x = 0$. Hence, we may assume that $x \neq 0$ and $y \neq 0$.

Suppose then that $x \geq 1$ and $y \geq 1$. Now we have either 
\[
I(u) = \{u + (0,-1), u + (-1,0), u + (1,1)\}
\] or 
\[
I(u) = \{u + (0,1), u + (1,0), u + (-1,-1)\}.
\]
In both cases, we obtain that $\bigcap_{c \in I(u)} N[c] = \{u\}$ and the condition is satisfied. The other (three) cases with $x \leq -1$ or $y \leq -1$ can be handled analogously. Thus, in conclusion, $C$ is a solid-locating-dominating code in $G$.
\end{proof}

\end{theorem}

\begin{figure}[H]
  \centering
    \begin{tikzpicture}
    \matrix[square matrix]
    {
    |[fill=darkgray]|$ $ & $ $ & $ $ & |[fill=darkgray]|$ $ & $ $ & $ $ & |[fill=darkgray]|$ $ & $ $ & $ $ & |[fill=darkgray]|$ $ & $ $ & $ $ & |[fill=darkgray]|$ $ \\
    $ $ & $ $ & |[fill=darkgray]|$ $ & $ $ & $ $ & |[fill=darkgray]|$ $ & $ $ & |[fill=darkgray]|$ $ & $ $ & $ $ & |[fill=darkgray]|$ $ & $ $ & $ $ \\
    $ $ & |[fill=darkgray]|$ $ & $ $ & $ $ & |[fill=darkgray]|$ $ & $ $ & $ $ & $ $ & |[fill=darkgray]|$ $ & $ $ & $ $ & |[fill=darkgray]|$ $ & $ $ \\
    |[fill=darkgray]|$ $ & $ $ & $ $ & |[fill=darkgray]|$ $ & $ $ & $ $ & |[fill=darkgray]|$ $ & $ $ & $ $ & |[fill=darkgray]|$ $ & $ $ & $ $ & |[fill=darkgray]|$ $ \\
    $ $ & $ $ & |[fill=darkgray]|$ $ & $ $ & $ $ & |[fill=darkgray]|$ $ & $ $ & |[fill=darkgray]|$ $ & $ $ & $ $ & |[fill=darkgray]|$ $ & $ $ & $ $ \\
    $ $ & |[fill=darkgray]|$ $ & $ $ & $ $ & |[fill=darkgray]|$ $ & $ $ & $ $ & $ $ & |[fill=darkgray]|$ $ & $ $ & $ $ & |[fill=darkgray]|$ $ & $ $ \\
    |[fill=darkgray]|$ $ & $ $ & $ $ & |[fill=darkgray]|$ $ & $ $ & $ $ & |[fill=darkgray]|$ $ & $ $ & $ $ & |[fill=darkgray]|$ $ & $ $ & $ $ & |[fill=darkgray]|$ $ \\
    $ $ & |[fill=darkgray]|$ $ & $ $ & $ $ & |[fill=darkgray]|$ $ & $ $ & $ $ & $ $ & |[fill=darkgray]|$ $ & $ $ & $ $ & |[fill=darkgray]|$ $ & $ $ \\
    $ $ & $ $ & |[fill=darkgray]|$ $ & $ $ & $ $ & |[fill=darkgray]|$ $ & $ $ & |[fill=darkgray]|$ $ & $ $ & $ $ & |[fill=darkgray]|$ $ & $ $ & $ $ \\
    |[fill=darkgray]|$ $ & $ $ & $ $ & |[fill=darkgray]|$ $ & $ $ & $ $ & |[fill=darkgray]|$ $ & $ $ & $ $ & |[fill=darkgray]|$ $ & $ $ & $ $ & |[fill=darkgray]|$ $ \\
    $ $ & |[fill=darkgray]|$ $ & $ $ & $ $ & |[fill=darkgray]|$ $ & $ $ & $ $ & $ $ & |[fill=darkgray]|$ $ & $ $ & $ $ & |[fill=darkgray]|$ $ & $ $ \\
    $ $ & $ $ & |[fill=darkgray]|$ $ & $ $ & $ $ & |[fill=darkgray]|$ $ & $ $ & |[fill=darkgray]|$ $ & $ $ & $ $ & |[fill=darkgray]|$ $ & $ $ & $ $ \\
    |[fill=darkgray]|$ $ & $ $ & $ $ & |[fill=darkgray]|$ $ & $ $ & $ $ & |[fill=darkgray]|$ $ & $ $ & $ $ & |[fill=darkgray]|$ $ & $ $ & $ $ & |[fill=darkgray]|$ $ \\
    };
    \end{tikzpicture}
\centering
\caption{
The darkened squares form a solid-locating-dominating code of density $\frac{1}{3}$ in the king grid.}
\label{DLD in king grid}
\end{figure}

Usually, the best known constructions for domination type codes in infinite grids are formed by a repetition of a finite pattern. However, this is not the case with the code $C$ of the previous theorem. 
Another observation is that the codeword $c = (0,0)$ has a special role as a sort of center of the code. In particular, the density of the code (or more precisely the ratio $|C \cap V_n|/|V_n|$) in the close proximity of $c$ is less than $1/3$. Consider now the lower bound on the density of a solid-locating-dominating code. Usually, the lower bounds are obtained by locally studying the symmetric difference of closed neighbourhoods of vertices or the domination properties of vertices (such as the concept of share~\cite{S:fault-tolerant} or the common technique used in the proof of Theorem~\ref{SLDkolmio}). However, in order to deal with the special type of codewords $c$, we develop a new technique of more global nature.
For this purpose, we first present the following lemma on a forbidden pattern of non-codewords.
\begin{lemma}\label{No T}
Let $G=(V,E)$ be the king grid and $C\subseteq V$ be a solid-locating-dominating code in $G$. Then $T=\{(i,j),(i,j+1),(i,j+2),(i+1,j+2),(i-1,j+2)\}$ and any formation obtained from $T$ by a rotation of $\pi/2$, $\pi$ or $3\pi/2$ radians around the origin contains a codeword of $C$.
\begin{proof}
Assume that the set $T=\{(i,j),(i,j+1),(i,j+2),(i+1,j+2),(i-1,j+2)\}$ contains no codewords of $C$. Then a contradiction with the definition follows since $I(i,j+1) \setminus I(i,j) = \emptyset$. The other cases obtained from $T$ by a rotation are proved analogously.
\end{proof}
\end{lemma}

In the following theorem, we prove that the solid-locating-dominating code of Theorem~\ref{ThmDLDKingConstruction} is optimal, i.e., there is no code with density smaller than $1/3$. The proof is based on the idea of studying one-way infinite strips of vertices of width $3$ and showing that the density of codewords in these strips is at least $1/3$.
\begin{theorem} \label{ThmKingLowerBound}
If $G=(V,E)$ is the king grid and $C\subseteq V$ is a solid-locating-dominating code in $G$, then the density $D(C) \geq \frac{1}{3}$.
\begin{proof}
Let $S^j$ be a subgraph of $G$ induced by the vertex set $V'_j=\left\{(x,y)\mid 1\leq x\leq 3, 1\leq y\leq j\right\}$. Recall first the definition $V_n = \{(x,y) \mid |x| \leq n, |y| \leq n\}$. Observe now that we may fit into the first quadrant $\{(x,y)\mid 1\leq x\leq n, 1\leq y\leq n\}$ of $V_n$ $\lfloor n/3 \rfloor$ graphs isomorphic to $S^n$. Similarly, the other three quadrants of $V_n$ can each contain $\lfloor n/3 \rfloor$ graphs isomorphic to $S^n$. Thus, in total, $4\lfloor n/3 \rfloor$ graphs isomorphic to $S_n$ can be fit into $V_n$.

%

Let $C$ be a solid-locating-dominating code in $G$. In the final part of the proof, we show that any subgraph of $G$ isomorphic to $S^n$ contains at least $n-3$ codewords. Assuming this is the case, the density of $C$ can be estimated as follows:
\[
D(C)= \limsup_{n \rightarrow \infty} \frac{|C \cap V_n|}{|V_n|} \geq \limsup_{n \rightarrow \infty} \frac{4\lfloor\frac{n}{3}\rfloor\cdot(n-3)}{(2n+1)^2} \geq \limsup_{n \rightarrow \infty} \frac{4(n-3)^2}{3(2n+1)^2} = \frac{1}{3} \text{.}
\]

It remains to be shown that any subgraph of $G$ isomorphic to $S^n$ contains at least $n-3$ codewords. By symmetry, it is enough to show that $|C\cap V'_n|\geq n-3$. In what follows, we consider more closely the number of codewords in a row $S_i = \{(j,i)\mid 1\leq j\leq3 \}$ of $V'_n$. For this purpose, the following set of rules for rearranging the codewords inside $V'_n$ are introduced:
%
\begin{itemize}[wide=0pt, leftmargin=\dimexpr\labelwidth + 2\labelsep\relax]
\item[\textbf{Rule} $1.1$:] If $S_i\cap C=\emptyset$, $1\leq i \leq n-1$ and $\{(1,i+1),(3,i+1) \}\subseteq C$, then one codeword is moved from $S_{i+1}$ to $S_i$. The rule is illustrated in Figure~\ref{R1}.
\item[\textbf{Rule} $1.2$:] If $S_i\cap C=\emptyset$, $2\leq i$ and $\{(1,i-1),(3,i-1) \}\subseteq C$, then one codeword is moved from $S_{i-1}$ to $S_i$. The rule can be viewed as a reflected version of Rule~1.1.
\item[\textbf{Rule} $2.1$:] If $S_i\cap C=\emptyset$, $2\leq i$ and $\{(1,i-1),(2,i-1) \}= C\cap S_{i-1}$, then one codeword is moved from $S_{i-1}$ to $S_i$. The rule is illustrated in Figure~\ref{R2}.
\item[\textbf{Rule} $2.2$:] If $S_i\cap C=\emptyset$, $2\leq i$ and $\{(2,i-1),(3,i-1) \}= C\cap S_{i-1}$, then one codeword is moved from $S_{i-1}$ to $S_i$. The rule can be viewed as a reflected version of Rule~2.1.
\item[\textbf{Rule} $3.1$:] If $S_i\cap C=\emptyset$, $3\leq i$, $S_{i-1}\cap C=\{(1,i-1)\}$ and $\{(2,i-2),(3,i-2)\}\subseteq S_{i-2}\cap C$, then one codeword is moved from $S_{i-2}$ to $S_i$. The rule is illustrated in Figure~\ref{R3}.
\item[\textbf{Rule} $3.2$:] If $S_i\cap C=\emptyset$, $3\leq i$, $S_{i-1}\cap C=\{(3,i-1)\}$ and $\{(2,i-2),(1,i-2)\}\subseteq S_{i-2}\cap C$, then one codeword is moved from $S_{i-2}$ to $S_i$. The rule can be viewed as a reflected version of Rule~3.1.
\item[\textbf{Rule} $4.1$:] If $S_i\cap C=\emptyset$, $3\leq i$, $S_{i-1}\cap C=\{(1,i-1)\}$ and $\{(1,i-2),(2,i-2)\}= S_{i-2}\cap C$, then one codeword is moved from $S_{i-2}$ to $S_i$. The rule is illustrated in Figure~\ref{R4}.
\item[\textbf{Rule} $4.2$:] If $S_i\cap C=\emptyset$, $3\leq i$, $S_{i-1}\cap C=\{(3,i-1)\}$ and $\{(2,i-2),(3,i-2)\}= S_{i-2}\cap C$, then  one codeword is moved from $S_{i-2}$ to $S_i$. The rule can be viewed as a reflected version of Rule~4.1.
\end{itemize}

\begin{figure}[H]
  \centering
  \begin{minipage}[b]{0.2\textwidth}
    \begin{tikzpicture}
    \matrix[square matrix]
    {
    $X $ & $? $ & $X $ \\
    $ $ & $ $ & $ $ \\
    };

            \coordinate(d1);
        \coordinate[below left = 0.3cm and 0.8cm of d1] (d2);
        \coordinate[above = 0.6 cm of d2] (d3);

 [above left = 1.8cm and 0.75cm of 1]

    \path[draw,thick]
    (d3) edge [->, bend right=80, looseness=1.5] node {} (d2)

 ;

\end{tikzpicture}
\centering
\caption{\newline Rule $1.1$\newline\newline}\label{R1}   \end{minipage}
  \hfill
  \begin{minipage}[b]{0.2\textwidth}
\begin{tikzpicture}
    \matrix[square matrix]
    {
    $ $ & $ $ & $ $ \\
    $X $ & $X $ & $ $ \\
    };

            \coordinate(d1);
        \coordinate[below left = 0.3cm and 0.8cm of d1] (d2);
        \coordinate[above = 0.6 cm of d2] (d3);

 [above left = 1.8cm and 0.75cm of 1]

    \path[draw,thick]
    (d2) edge [->, bend left=80, looseness=1.5] node {} (d3)

 ;
\end{tikzpicture}
\centering

\caption{\newline Rule $2.1$\newline\newline}\label{R2}
  \end{minipage}
  \hfill
  \begin{minipage}[b]{0.2\textwidth}
\begin{tikzpicture}
    \matrix[square matrix]
    {
    $ $ & $ $ & $ $ \\
    $X $ & $ $ & $ $ \\
    $? $ & $X $ & $X $ \\
    };

            \coordinate(d1);
        \coordinate[below left = 0.5cm and 0.8cm of d1] (d2);
        \coordinate[above of=d2] (d3);

 [above left = 1.8cm and 0.75cm of 1]

    \path[draw,thick]
    (d2) edge [->, bend left=80, looseness=1.5] node {} (d3)

 ;
\end{tikzpicture}
\centering

\caption{\newline Rule $3.1$\newline\newline}\label{R3}
  \end{minipage}
    \hfill
  \begin{minipage}[b]{0.2\textwidth}
\begin{tikzpicture}
    \matrix[square matrix]
    {
    $ $ & $ $ & $ $ \\
    $X $ & $ $ & $ $ \\
    $X $ & $X $ & $ $ \\
    };
        \coordinate(d1);
        \coordinate[below left = 0.5cm and 0.8cm of d1] (d2);
        \coordinate[above of=d2] (d3);

 [above left = 1.8cm and 0.75cm of 1]

    \path[draw,thick]
    (d2) edge [->, bend left=80, looseness=1.5] node {} (d3)

 ;
\end{tikzpicture}
\centering

\caption{\newline Rule $4.1$\newline\newline}\label{R4}
  \end{minipage}
\end{figure}

\vspace*{-3ex}
Denote the code obtained after simultaneously applying the previous rules by $C'$. Notice that the rearrangement $C'$ of $C$ is not completely determined by the previous rules and that this is not actually needed as in the following we are only interested on the number of codewords in the rows of $V'_n$. In other words, when a codeword is moved from a row we can choose any of the codewords and move it to replace any non-codeword of the target row. In what follows, we show that each row which has given away codewords still contains at least one and each row which originally did not contain any codeword has received at least one except possibly the rows $S_1$, $S_2$ and $S_n$.


We immediately notice that the rules move codewords only from the rows with at least two codewords. Each type of row with at least two codewords is examined as follows:
\begin{itemize}[wide=0pt, leftmargin=\dimexpr\labelwidth + 2\labelsep\relax]
\item $C\cap S_i = \{(j,i)\mid 1\leq j\leq3\}$: Rules $1.1, 1.2, 3.1$ and $3.2$ can be applied on rows with three codewords. Among these, Rules $3.1$ and $3.2$ cannot be applied at the same time and Rule $1.2$ cannot be applied together with Rules $3.1$ or $3.2$. Hence, we apply at most two rules on a row with three codewords and that row has at least one codeword left in the code $C'$.
\item $C\cap S_i = \{(j,i)\mid 1\leq j\leq2\}$: Rules $2.1,3.2$ and $4.1$ can be applied on this types of rows. We cannot apply Rule $2.1$ at the same time as $3.2$ or $4.1$ since $2.1$ requires that $S_{i+1}\cap C=\emptyset$ and Rules $3.2$ and $4.1$ require that $|S_{i+1}\cap C|=1$. Furthermore, we cannot apply Rules $3.2$ and $4.1$ at the same time since they require the codeword on the row $S_{i+1}$ to locate at different places. Hence, $C'$ is left with at least one codeword.
\item $C\cap S_i = \{(j,i)\mid 2\leq j\leq3\}$: This case is symmetrical to the previous one (now the rules to be considered are 2.2, 3.1 and 4.2).
\item $C\cap S_i = \{(j,i)\mid j\neq 2\}$: We can only apply Rules $1.1$ and $1.2$ on these types of rows and both of them only when $i\geq 2$. However, if both of the rules are used, then $C \cap S_{i-1} = C \cap S_{i+1} = \emptyset$ and a contradiction with Lemma~\ref{No T} follows. Hence, at most one rule is used and $|C' \cap S_i| \geq 1$.
\end{itemize}

Let us then show that we have $|C'\cap S_i|\geq1$ for each $i$ such that $C\cap S_i=\emptyset$ and $3 \leq i \leq n-1$. In the following cases, we assume that $S_{i}\cap C=\emptyset$ and the cases are categorized by considering the different formations of the row $S_{i-1}$.
\begin{itemize}[wide=0pt, leftmargin=\dimexpr\labelwidth + 2\labelsep\relax]
\item $S_{i-1}\cap C=\emptyset$: Considering different orientations and positions of the formation $T$ in Lemma~\ref{No T}, we have $S_{i+1}\subseteq C$. Hence, due to Rule $1.1$, one codeword from $S_{i+1}$ is moved to $S_i$ and we obtain $|C'\cap S_{i}|\geq1$.
\item $S_{i-1}\cap C=\{(1,i-1)\}$: By Lemma~\ref{No T}, we have $(2,i-2)\in C$. Notice that if $(1,i-2)$ and $(3,i-2)$ do not belong to $C$, then a contradiction with the definition of solid-locating-dominating codes follows since we have $I(2,i-1) \subseteq I(1,i-2)$ for non-codewords $(2,i-1)$ and $(1,i-2)$. Hence, at least one of the vertices $(1,i-2)$ and $(3,i-2)$ belongs to $C$. Therefore, either Rule $3.1$ or $4.1$ can be applied (to the row $S_{i-2}$) and we have $|C'\cap S_{i}|\geq 1$.
\item $S_{i-1}\cap C=\{(3,i-1)\}$: This case is symmetrical to the previous one. Here we just use either Rule $3.2$ or $4.2$.
\item $S_{i-1}\cap C=\{(2,i-1)\}$: By Lemma~\ref{No T}, we have $\{(1,i+1),(3,i+1)\}\subseteq C$. Hence, due to Rule $1.1$, we have $|C'\cap S_{i}|\geq1$.
\item $S_{i-1}\cap C=\{(1,i-1),(2,i-1)\}$: Due to Rule $2.1$, we have $|C'\cap S_{i}|\geq1$.
\item $S_{i-1}\cap C=\{(2,i-1),(3,i-1)\}$: Due to Rule $2.2$, we have $|C'\cap S_{i}|\geq1$.
\item $S_{i-1}\cap C=\{(1,i-1),(3,i-1)\}$: Due to Rule $1.2$, we have $|C'\cap S_{i}|\geq1$.
\item $S_{i-1}\cap C=\{(1,i-1),(2,i-1),(3,i-1)\}$: Due to Rule $1.2$, we have $|C'\cap S_{i}|\geq1$.
\end{itemize}
Thus, in conclusion, we have shown that for $3 \leq i \leq n-1$ we have $|C' \cap S_i| \geq 1$. Therefore, as the rules rearrange codewords only inside $V'_n$, we have $|C \cap V'_n| \geq |C' \cap V'_n| \geq n-3$. This concludes the proof of the lower bound $D(C) \geq 1/3$.
%
\end{proof}
\end{theorem}

In the previous theorems, we have shown that the density of an optimal solid-locating-dominating code in the king grid is $1/3$. Recall that a self-locating-dominating code is always solid-locating-dominating. Hence, by the previous lower bound, we also know that there exists no self-locating-dominating code in the king grid with density smaller than $1/3$. However, the construction given for the solid-location-domination does not work for self-location-domination. For example, we have $I(2,0) = \{(2,-1), (2,1), (3,0)\}$ and $N[(2,-1)] \cap N[(2,1)] \cap N[(3,0)] = \{(2,0), (3,0)\}$ contradicting with the definition of self-locating-dominating codes (see Figure~\ref{DLD in king grid}). In the following theorem, we present a self-locating-dominating code in the king grid with the density $1/3$. 
 Theorems~\ref{ThmDLDKingConstruction},~\ref{ThmKingLowerBound} and~\ref{thmKingSLDConst} imply that while the optimal density for both self- and solid-locating-dominating codes in the infinite king grid is $1/3$, there exist solid-locating-dominating codes which are not self-locating-dominating although each self-locating-dominating code is always solid-locating-dominating.

\begin{theorem}\label{thmKingSLDConst}
Let $G=(V,E)$ be the king grid. The code
\[
C =  \left\{(x,y)\in \Z^2\mid x - y \equiv 0 \pmod 3\right\}
\]
is self-locating-dominating in $G$ and its density is $1/3$.
\begin{proof}
The density $D(C) = 1/3$ since in each row every third vertex is a codeword. Furthermore, $C$ is a self-locating-dominating code since each non-codeword $v$ is covered either by the set of three codewords $\{v+(1,0),v+(0,-1),v+(-1,1)\}$ or $\{v+(-1,0),v+(0,1),v+(1,-1)\}$, and in both cases the closed neighbourhoods of the codewords intersect uniquely in the vertex $v$.
\end{proof}
\end{theorem}


\section{Direct product of complete graphs}\label{secDirProd}

A graph is called a \emph{complete graph} on $q$ vertices, denoted by $K_q$, if each pair of vertices of the graph is adjacent. The vertex set $V(K_q)$ is denoted by $\{1,2, \ldots, q\}$. The \emph{Cartesian product} of two graphs $G_1=(V_1,E_1)$ and $G_2=(V_2,E_2)$ is defined as $G_1\square G_2=(V_1\times V_2,E)$, where $E$ is a set of edges such that $(u_1,u_2)(v_1,v_2)\in E$ if and only if $u_1=v_1$ and $u_2v_2\in E_2$, or $u_2=v_2$ and $u_1v_1\in E_1$. The \emph{direct product} of two graphs $G_1$ and $G_2$ is defined as $G_1\times G_2=(V_1\times V_2, E)$, where $E=\{(u_1,u_2)(v_1,v_2)\mid u_1v_1\in E_1\text{ and } u_2v_2\in E_2\}$. A \emph{complement} of a graph $G = (V,E)$ is the graph $\overline{G} = (V,E')$ with the edge set $E'$ being such that $uv \in E'$ if and only if $uv \notin E$.

In this section, we give optimal locating-dominating, self-locating-dominating and solid-locating-dominating codes in the direct product $K_n\times K_m$, where $2\leq n\leq m$. For location-domination and solid-location-domination, the results heavily depend on the exact values of $\LD(K_n \square K_m)$ and $\GSLD(K_n \square K_m)$, which have been determined in \cite{JLLrntcld}. In the graphs $K_n\times K_m$ and $K_n \square K_m$, the $j$th \emph{row} (of $V(K_n)\times V(K_m)$) is denoted by $R_j$ and it consists of the vertices $(1,j), (2,j), \ldots, (n,j)$. Analogously, the $i$th \emph{column} is denoted by $P_i$ and it consists of the vertices $(i,1), (i,2), \ldots, (i,m)$. Now we are ready to present the following observations:
\begin{itemize}
\item In the Cartesian product $K_n \square K_m$, the closed neighbourhood $N[(i,j)] = N[i,j]$ consists of the row $R_j$ and the column $P_i$. Therefore, as the closed neighbourhood of a vertex resembles the movements of a rook in a chessboard, $K_n \square K_m$ is also sometimes called the \emph{rook's graph}.
\item In the direct product $K_n \times K_m$, we have $N((i,j)) = N(i,j) = V(K_n \square K_m) \setminus (R_j \cup P_i)$.
\end{itemize}
Due to the previous observations, we know that $\overline{K_n\square K_m}=K_n\times K_m$. 

Recall that identification is a topic closely related to the various location-domination type problems. Previously, in~\cite{WashThesis}, the identifying codes have been studied in the direct product $K_n\times K_m$ of complete graphs by Rall and Wash. More precisely, they determined the exact values of $\ID(K_n\times K_m)$ for all $m$ and $n$.


In what follows, we determine the exact values of $\LD(K_n \times K_m)$ for all $m$ and $n$. For this purpose, we first present the following result concerning location-domination in the Cartesian product $K_n\square K_m$ of complete graphs given in~\cite{JLLrntcld}.
\begin{theorem}[\cite{JLLrntcld},~Theorem 14]\label{cartesian LD}
Let $m$ and $n$ be integers such that $2 \leq n\leq m$. Now we have
$$\LD(K_n \square K_m)=\begin{cases}
 m-1, & 2n\leq m, \
\cr \left\lceil \frac{2n+2m}{3}\right\rceil-1, & n\leq m\leq 2n-1.
\end{cases}$$
\end{theorem}

There is a strong connection between the values of $\LD(K_n \square K_m)$ and $\LD(K_n \times K_m)$ as explained in the following. In~\cite{complementLD}, it has been shown that $|\gamma^{LD}(G)-\LD(\overline{G})|\leq1$. Therefore, as $\overline{K_n\times K_m}=K_n\square K_m$, we obtain that $\LD(K_n\square K_m)-1\leq \LD(K_n\times K_m)\leq \LD(K_n\square K_m)+1$. This result is further sharpened in the following lemma.
\begin{lemma}\label{LD Kq x Kp}
For $2\leq n\leq m$ and $(n,m) \neq (2,4)$, we have $$\gamma^{LD}(K_n\square K_m)-1\leq \gamma^{LD}(K_n\times K_m)\leq \gamma^{LD}(K_n\square K_m).$$
If $\LD(K_n\times K_m) = \LD(K_n\square K_m) -1$, then every optimal locating-dominating code $C$ in $K_n\times K_m$ has a non-codeword $v$ such that $I(v)=C$.
\begin{proof}
First denote $G=K_n\square K_m$ and $H=K_n\times K_m$. The lower bound of the claim is immediate by the result preceding the lemma. For the upper bound, let $C$ be an optimal locating-dominating code in $G$. The code $C$ can also be viewed as a code in $H$. If we have $I(H;u)=I(H;v)$ for some non-codewords $u$ and $v$, then a contradiction follows since $I(G;u)=C\setminus I(H;u)=C\setminus I(H;v)=I(G;v)$. Hence, we have $I(H;u) \neq I(H;v)$ for all distinct non-codewords $u$ and $v$. Moreover, if $I(G;v)\neq C$ for each non-codeword $v$, then we also have $I(H;v)\neq \emptyset$, and the upper bound follows since $C$ is a locating-dominating code in $H$.

Hence, we may assume that $I(G;v)=C$ for some non-codeword $v$. This implies that $C\subseteq P_i\cup R_j$ for some $i,j$. There exists at most one non-codeword in $P_i \setminus \{v\}$ since otherwise there are at least two non-codewords with the same $I$-set. Similarly, there exists at most one non-codeword in $R_j \setminus \{v\}$. Furthermore, if both $P_i \setminus \{v\}$ and $R_j \setminus \{v\}$ contain a non-codeword, then there exists a vertex with an empty $I$-set. Thus, in conclusion, there exists at most two non-codewords in $P_i\cup R_j$ and, hence, we have $|C|\geq n+m-3$. Dividing into the following cases depending on $n$ and $m$, we next show that $|C| \geq n+m-3>\LD(G)$ in majority of the cases of the lemma: 
\begin{itemize}
\item If $n \geq 3$ and $m \geq 2n$, then we have $\LD(G) = m-1 < n+m-3 \leq |C|$ (by Theorem~\ref{cartesian LD}).
\item If $n \geq 4$, $n \leq m \leq 2n-1$ and $(n,m) \neq (4,4)$, then $\LD(G) = \lceil 2(n+m)/3 \rceil - 1 < n+m-3 \leq |C|$ (by Theorem~\ref{cartesian LD}).
\end{itemize}
Thus, if $n \geq 3$ and $m \geq 2n$, or $n \geq 4$, $n \leq m \leq 2n-1$ and $(n,m) \neq (4,4)$, then a contradiction with the optimality of $C$ follows. Hence, in these cases, we have $\gamma^{LD}(H)\leq \gamma^{LD}(G)$.

The rest of the cases are covered in the following:
\begin{itemize}
\item If $n = 2$ and $2 \leq m \leq 3$, then $C = P_1$ is an optimal locating-dominating code in $G$ with the property that for any non-codeword $v$ we have $I(G; v) \neq C$. Similarly, if $n = 2$ and $m \geq 5$, then $C = \{(2,1),(2,2)\}\cup P_1\setminus\{(1,i)\mid i\leq3\}$ is an optimal locating-dominating code in $G$ with the property that for no vertex $v$ we have $I(G; v) = C$. Thus, in both cases, the code $C$ is also locating-dominating in $H$ by the first paragraph of the proof.
\item If $n = m = 3$, then $C = \{(1,1),(1,2),(2,1)\}$ is a locating-dominating code in $H$ with $\gamma^{LD}(G) = 3$ codewords.
\item If $n = 3$ and $4 \leq m \leq 5$ or $(n,m) = (4,4)$, then $\{(1,1),(1,3),(2,2),(2,4)\}$, $\{(1,1), (1,3),\linebreak (2,2),(2,4),(3,5)\}$ and $\{(1,1), (1,3), (2,2), (2,4), (3,1)\}$ obtained from the proof of \cite[Theorem 14]{JLLrntcld} are optimal locating-dominating codes in $K_3\square K_4$, $K_3\square K_5$ and $K_4\square K_4$, respectively. Therefore, since there does not exist a non-codeword covering all the codewords (in the Cartesian product) in any of the cases, the codes are also locating-dominating in $K_3\times K_4$, $K_3\times K_5$ and $K_4\times K_4$ (by the first paragraph of the proof), respectively.
\end{itemize}

Let then $C'$ be a locating-dominating code in $H$. Similarly as above, we get that if $I(H;v)\neq C'$ for each non-codeword $v$, then $C'$ is also a locating-dominating code in $G$. Therefore, if $\LD(H) = \LD(G) -1$, then there exist a non-codeword $v$ such that $I(H;v) = C'$. 
Thus, the last claim of the lemma follows.
\end{proof}
\end{lemma}

Now with the help of the previous lemma and Theorem~\ref{cartesian LD}, we determine the exact values of $\LD(K_m\times K_n)$ in the following theorem.
\begin{figure}[!tbp]
  \centering
\captionsetup{justification=centering}
\begin{tikzpicture}
    \matrix[big square matrix]
    {
    $ $ & $ $ & $ $ & $ $ & $ $ & $ $ & $ $ & $ $ & $ $ & $ $ \\
    $ $ & $ $ & $ $ & |[fill=darkgray]|$ $ & $ B_1$ & $ B_1$ & $ $ & $ $ & $ $ & $ $ \\
    $ $ & $ $ & $ $ & $  B_1$ & |[fill=darkgray]|$ $ & $ B_1$ & $ $ & $ $ & $ $ & $ $ \\
    $ $ & $ $ & $ $ & $ B_1$ & $ B_1$ & |[fill=darkgray]|$ $ & $ $ & $ $ & $ $ & $ $ \\
    $ $ & $ $ & $ $ & $ B_1$ & $ B_1$ & |[fill=darkgray]|$ $ & $ $ & $ $ & $ $ & $ $ \\
    $ $ & $ $ & $ $ & $ B_1$ & |[fill=darkgray]|$ $ & $ B_1$ & $ $ & $ $ & $ $ & $ $ \\
    $ $ & $ $ & $ $ & |[fill=darkgray]|$ $ & $ B_1$ & $ B_1$ & $ $ & $ $ & $ $ & $ $ \\
    $ B_2$ & $ B_2$ & |[fill=darkgray]|$ $ & $ B_1B_2$ & $ B_1B_2$ & $ B_1B_2$ & $ B_2$ & $ B_2$ & |[fill=darkgray]|$ $ & $ $ \\
    $ B_2$ & |[fill=darkgray]|$ $ & $ B_2$ & $ B_1B_2$ & $ B_1B_2$ & $ B_1B_2$ & $ B_2$ & |[fill=darkgray]|$ $ & $ B_2$ & $ $ \\
    |[fill=darkgray]|$ $ & $ B_2$ & $ B_2$ & $ B_1B_2$ & $ B_1B_2$ & $ B_1B_2$ & |[fill=darkgray]|$ $ & $ B_2$ & $ B_2$ & $ $ \\
    };
\end{tikzpicture}
\centering
\captionsetup{justification=raggedright}
\caption{Optimal locating-dominating code for $K_{10}\times K_{10}$. Dark boxes are codewords.}\label{10x10 LD}
\centering
\end{figure}
\begin{theorem}\label{direct LD}
For $2\leq n\leq m$ we have $$\LD(K_n\times K_m)=\begin{cases}m-1,&2n\leq m \text{ and } (n,m)\neq (2,4),\\
\left\lceil\frac{2n+2m-1}{3}\right\rceil-1,& 2<n\leq m<2n \text{ and } (m,n) \neq (4,4)\\
m,& n=2, m\leq4,\\
5,& n=4,m=4.\end{cases}$$
\begin{proof}
Let $C$ be a locating-dominating code in $K_n\times K_m$. We cannot have $R_i\cap C=R_j\cap C=\emptyset$ for $i\neq j$ since otherwise, for example, $I(C;(1,i)) = I(C;(1,j))$. Similarly, there exists at most one column without codewords of $C$. 
Thus, we have $\LD(K_n\times K_m)\geq m-1$. Therefore, if $m\geq2n$ and $(n,m)\neq(2,4)$, then by the previous lemma we have $m-1 \leq \LD(K_n\times K_m) \leq \LD(K_n\square K_m) = m-1$, i.e., $\LD(K_n\times K_m) = m-1$. 

Assume then that $2< n\leq m\leq 2n-1$ and $n+m \equiv 0, 1 \! \pmod{3}$. In what follows, we show that now $|C| \geq \LD(K_n\square K_m)$. By the previous lemma, we know that if there is no non-codeword $u$ such that $I(K_n\times K_m, C; u) = C$, i.e., there does not exist a row and column without codewords, then $|C| =\LD(K_n\square K_m)$. 
Hence, we may now assume that there exist a row and a column without codewords. Without loss of generality, we may assume that they are $P_n$ and $R_m$. Observe that $C$ can now also be viewed as a code in $K_{n-1}\square K_{m-1}$ and that $C$ is locating-dominating in $K_{n-1}\square K_{m-1}$ with the following additional properties: (i) each column has at least one codeword, (ii) each row has at least one codeword and (iii) no codeword $(i,j) \in C$ is such that $(P_i \cup R_j) \cap C = \{(i,j)\}$, i.e., no codeword of $C$ is isolated. Indeed, the properties~(i) and (ii) follow immediately by the first paragraph of the proof and if $(i,j) \in C$ is a codeword violating the property~(iii), then we have $I(K_n\times K_m,(n,j))=I(K_n\times K_m;(i,m))=C\setminus \{(i,j)\}$ (a contradiction). Now we are ready to prove a lower bound on $|C|$ as in \cite[Theorem 14]{JLLrntcld}. Denote the number of columns and rows with exactly one codeword in $K_n\square K_m$ by $s_p$ and $s_r$, respectively. Now we obtain that $|C| \geq s_p + 2(n-1-s_p) = 2(n-1)-s_p$ and $|C| \geq s_r + 2(m-1-s_r) = 2(m-1)-s_r$ (by the properties~(i) and (ii)). This further implies that $s_p \geq 2(n-1) - |C|$ and $s_r \geq 2(m-1) - |C|$. By the property~(iii), we now obtain that $|C|\geq s_p+s_r\geq 2(n-1)+2(m-1)-2|C|$. Thus, we have $|C|\geq \left\lceil (2m+2n-1)/3 \right\rceil-1$. Hence, as $n+m \equiv 0, 1 \! \pmod{3}$, we have $|C|\geq \left\lceil (2m+2n-1)/3 \right\rceil-1 = \left\lceil (2m+2n)/3 \right\rceil -1 = \LD(K_n\square K_m)$. Thus, by the upper bound of the previous lemma, we obtain that $\LD(K_n\times K_m) = \LD(K_n\square K_m)$ if $2< n\leq m\leq 2n-1$ and $n+m \equiv 0, 1 \! \pmod{3}$.


Assume then that $2< n\leq m\leq 2n-1$, $n+m \equiv 2 \! \pmod{3}$ and $(n,m) \neq (4,4)$. In what follows, we show that the lower bound of Lemma~\ref{LD Kq x Kp} is attained, i.e., $\LD(K_n\times K_m) = \LD(K_n\square K_m) - 1$. Denote $n'=n-1$ and $m'=m-1$ and observe that $n'+m'$ is divisible by three. Let $C'=A_1\cup A_2\cup A_3$ be a code in $K_n \times K_m$ with
\begin{align*}
A_1&=\left\{(i,i)\mid 1\leq i\leq \frac{n'+m'}{3}\right\},\\
A_2&=\left\{(j,i)\mid \frac{n'+m'}{3}+1\leq i\leq m', j=2\frac{n'+m'}{3}+1-i\right\}\text{ and }\\
A_3&=\left\{(j,i)\mid 1\leq i\leq \frac{2n'-m'}{3}, j=i+\frac{n'+m'}{3}\right\}.
\end{align*}
The code $C'$ is illustrated in Figure~\ref{10x10 LD}. By straightforward counting , we get $|C'| = |A_1|+|A_2|+|A_3| = m' + \frac{2n'-m'}{3}=\frac{2n'+2m'}{3}=\frac{2n+2m-1}{3}-1=\LD(K_n\square K_m)-1$. In what follows, we first show that $C'$ is almost a locating-dominating code in $K_n\square K_m$ with the exception that $I(C';(n,m)) = \emptyset$.

Denote the sets of non-codewords $(j,i)$ with $(2n'+2m')/3-m'+1\leq j\leq (n'+m')/3$ and $i\leq (2n'+2m')/3-m'$ by $B_1$ and $B_2$, respectively. It is straightforward to verify that each non-codeword $u \in B_1 \cup B_2$ has at least three codewords in $I(K_n\square K_m,C';u)$  and the codewords of $I(K_n\square K_m,C';u)$ do not lie on a single row or column. This implies that $\bigcap_{c \in I(C';u)} N[c] = \{u \}$ for any $u \in B_1 \cup B_2$, i.e., there is no other vertex containing $I(C';u)$ in its $I$-set. Thus, each non-codeword in $B_1 \cup B_2$ has a unique nonempty $I$-set. Consider then a non-codeword $v=(j,i)$ with $i> (2n'+2m')/3-m'$ and $j<(2n'+2m')/3-m+1'$. By the construction of $C'$, we have $|I(C';v)| = 2$. Now there exists a codeword $(j,j)\in I(v)$ since $j\leq (2n'+2m')/3-m'$. Furthermore, there exists a codeword $c \in I(j,j)\cap A_3$. Hence, if there exists a non-codeword $w$ such that $I(C';v) = I(C';w)$, then $w \in B_2$ and a contradiction follows as $|I(C';w)| \geq 3$. Thus, the $I$-set of $v$ is nonempty and unique. Similarly, it can be shown that $I(C';(j,i))$ is nonempty and unique for $i>(2n'+2m')/3-m'$ and $j>(n'+m')/3$.

Consider then non-codewords $u = (j,m)$ and  $v = (n, i)$ with $1 \leq j \leq n-1$ and $1 \leq i \leq m-1$. We immediately obtain that $I(C';(j,m)) = P_j \cap C'$ and $I(C';(n,i)) = R_i \cap C'$. These $I$-sets are nonempty since each row and column contains a codeword. These $I$-sets are also different from the ones of non-codewords inside $K_{n'} \square K_{m'}$ which contain at least two codewords in different rows and columns. It is also impossible to have $I(C';u) = I(C';v)$ since each codeword has another one in the same row or column. Thus, $u$ and $v$ have nonempty and unique $I$-sets. Thus, in conclusion, we have shown that $I(C';u)$ is nonempty and unique for all non-codewords $u$ in $K_{n} \square K_{m}$ except $(n,m)$ (for which we have $I(C';(n,m)) = \emptyset$). Furthermore, there does not exist a non-codeword $v$ such that $I(C';v) = C'$. Therefore, as in the proof of Lemma~\ref{LD Kq x Kp}, we obtain that $C'$ is a locating-dominating code in $K_{n} \times K_{m}$. Thus, we have $\LD(K_n\times K_m) = \LD(K_n\square K_m) - 1$.

Now majority of the cases have been considered, and we only have some special cases left. Concluding the proof, these cases are solved as follows:
\begin{itemize}
\item Assume that $n=2$ and $m\leq 4$. It is easy to see that $C=P_1$ is a locating-dominating code in $K_n\times K_m$. For the lower bound, first recall that $K_n\times K_m$ has at most one row without codewords (by the first paragraph of the proof). Therefore, if $C$ is a locating-dominating code in $K_n\times K_m$ with $|C| \leq m-1$, then all the codewords lie on different rows. Hence, in all the cases, there exist a non-codeword with an empty $I$-set. Thus, we have $\LD(K_n\times K_m) = m$.
\item Assume that $n = m = 4$. By Lemma~\ref{LD Kq x Kp}, we immediately have $4 \leq \LD(K_4\times K_4) \leq 5$. Let $C$ be a locating-dominating code in $K_4\times K_4$. As in the second paragraph of the proof, it can be shown that either $|C| \geq \LD(K_4\square K_4) = 5$ (and we are done), or $C$ is locating-dominating in $K_3\times K_3$ with the additional properties~(i), (ii) and (iii). In the latter case, due to (i), (ii) and (iii), there exist a row and a column of $K_3\times K_3$ with two codewords such that their intersection is a non-codeword $u$. Hence, a contradiction follows since $I(K_4\times K_4, C; u) = \emptyset$. Thus, we have $\LD(K_4\times K_4) = 5$.
\end{itemize}
\end{proof}
\end{theorem}

Let us next briefly consider solid-location-domination. The following result has been shown in~\cite{JLLrntcld}.
\begin{theorem}[\cite{JLLrntcld}]\label{cartesian DLD}
For all integers $m$ and $n$ such that $m \geq n\geq 1$, we have
$$\GSLD(K_n \square K_m)=\begin{cases}
 m, & 4\leq 2n\leq m \textnormal{ or } n= 2, \
\cr 2n, & 2<n < m <2n,
\cr 2n-1, & 2<m=n.
\end{cases}$$
\end{theorem}

In the following theorem, we show that the cardinalities of optimal solid-locating-dominating codes are same for $K_n\times K_m$ and $K_n\square K_m$.
\begin{theorem}
For all integers $m$ and $n$ such that $m \geq n\geq 2$, we have
$$\GSLD(K_n\times K_m)=\GSLD(K_n\square K_m).$$
\begin{proof}
By \cite[Theorem 21]{SLDDLDgraafit}, we have $\GSLD(G) = \GSLD(\overline{G})$ if $G$ is not a discrete or a complete graph. Therefore, as this is the case for $G = K_n\times K_m$, we have $\GSLD(K_n\times K_m) = \GSLD(\overline{K_n\times K_m}) = \GSLD(K_n\square K_m)$.
\end{proof}
\end{theorem}

Let us then consider self-location-domination. Unlike location-domination~\cite[Theorem 7]{complementLD} and solid-location-domination~\cite[Theorem 21]{SLDDLDgraafit}, the optimal cardinality of a self-locating-dominating code in $G$ does not depend on the one of the complement graph $\overline{G}$. In the following theorem, we first give the result presented in~\cite{JLLrntcld} regarding $\SLD(K_n \square K_m)$.
\begin{theorem}[\cite{JLLrntcld}]\label{SLD}
For all integers $m$ and $n$ such that $m \geq n\geq 2$, we have
$$\SLD(K_n \square K_m)=\begin{cases}
 m, & 2n\leq m, \
\cr 2n, & 2\leq n<m<2n,
\cr 2n-1, & 2<m=n,
\cr 4, & n=m=2.
\end{cases}$$\end{theorem}

In the following theorem, we determine the exact values of $\SLD(K_n \times K_m)$ for all values of $m$ and $n$. Notice that $\SLD(K_n\square K_m)=\SLD(K_n\times K_m)$ if and only if $n=m$, $m=n+1>3$, or $n=2$ and $m\geq4$.
\begin{theorem}
For all integers $m$ and $n$ such that $m \geq n\geq 2$, we have
$$\SLD(K_n\times K_m)=\begin{cases}m+n-1,&n>2,\\
m,& n=2, m>2,\\
4,&n=m=2.\end{cases}$$
\begin{proof}
Let $C$ be a self-locating-dominating code in $K_n\times K_m$. Notice first that if $n = m =2$, then $K_2\times K_2$ is isomorphic to a forest of two paths of length two and, therefore, $\SLD(K_2\times K_2) = 4$. Hence, we may assume that $(n,m) \neq (2,2)$. Observe then that if a column $P_i$ contains no codewords, i.e., $P_i\cap C=\emptyset$, then $C = V \setminus P_i$. Indeed, for any vertices $(i,j) \in P_i$ and $(h,j) \in V$ with $i \neq h$, we have $I(h,j)\subseteq I(i,j)$ and the claim $C = V \setminus P_i$ follows by Theorem~\ref{ThmCharacterizationSLDILD}. Analogously, it can be shown that if $R_i\cap C=\emptyset$, then $C = V \setminus R_i$. Suppose now that $n = 2$ and $m > 2$. If each row contains a codeword, then we immediately have $|C| \geq m$. Otherwise, there exists a row without codewords and, by the previous observation, we have $|C| \geq 2m - 2 \geq m$. Hence, we obtain that $|C| \geq m$. Furthermore, $P_1$ is a self-locating-dominating code in $K_2\times K_m$ with $m$ codewords. Thus, in conclusion, we have $\SLD(K_2\times K_m) = m$.

Assume that $n > 2$. By the previous observations, we know that if there exists a row or a column without codewords, then $|C| \geq \min\{mn-m, mn-n\} = mn-m \geq m+n-1$. Hence, we may assume that each row and column contains a codeword of $C$. Furthermore, if each row contains at least $2$ codewords, then $|C| \geq 2m \geq m+n-1$. Hence, we may assume that there exists a row $R_i$ with exactly one codeword, i.e., $R_i\cap C = \{(j,i) \}$ for some $j$. Hence, as $I(j,h)\subseteq I(j,i)$ for any $h \neq i$, we have $P_j\subseteq C$. Therefore, as each column different from $P_j$ also contains a codeword, we obtain that $|C| \geq m+n-1$. Thus, we have $\SLD(K_n\times K_m) \geq m+n-1$. Finally, this lower bound can be attained with a code $C'=\{(i,j)\mid i=1\text{ or } j=1\}$. Indeed, for any $i,j > 1$, we have $I(1,1)=\{(1,1)\}$, $I(1,j)=\{(1,j)\}\cup (R_1\setminus\{(1,1)\})$, $I(j,1)=\{(j,1)\}\cup (P_1\setminus\{(1,1)\})$ and $I(i,j)= C' \setminus \{(1,j),(i,1)\}$. Therefore, we have $I(v)\not\subseteq I(u)$ for any vertex $u$ and non-codeword $v$. Thus, by Theorem~\ref{ThmCharacterizationSLDILD}, $C'$ is a self-locating-dominating code in $K_n\times K_m$, and we have $\SLD(K_n\times K_m) = n + m - 1$.
\end{proof}
\end{theorem}

\section{On certain type of Hamming graphs}\label{Hamming}

The Cartesian product $K_q\square K_q\square\cdots \square K_q$ of $n$ copies of $K_q$ is denoted by $K_q^n$ and called a Hamming graph. Goddard and Wash  \cite{GWIDcpg} studied identification in the case of $K_q^n$ and they, in particular, bounded the cardinality of an optimal identifying code to $q^2-q\sqrt{q}\leq \gamma^{ID}(K_q^3)\leq q^2$. In \cite{CubeCon}, we further improved this bound to $q^2-\frac{3}{2}q\leq \gamma^{ID}(K_q^3)\leq q^2-4^{t-1}$ where $2\cdot 4^t\leq q \leq 2\cdot 4^{t+1}-1$ or $q=4^t$, and we also showed that $\gamma^{SLD}(K^3_q)=q^2$. In this section, we show that also $\GSLD(K^3_q)=q^2$.

The following lemma is presented as Exercise 1.12 in \cite{DomInGraphs}.
\begin{lemma}[\cite{DomInGraphs}]\label{dominationNumber}
For each positive integer $q$, we have $$\gamma(K_q\square K_q)=q.$$
\end{lemma}

In the following we present some terminology and notations we use. More information about them and their usefulness can be found in \cite{CubeCon}.

\begin{itemize}

\item The \textit{pipe} $P^i(a,b)\subseteq V(K_q^3)$ is a set of vertices fixing all but the $i$th coordinate which varies between $1$ and $q$. The fixed coordinates are $a$ and $b$ where $a$ is the value of left fixed coordinate in the representation $(x,y,z)$. For example $P^3(a,b)=\{(a,b,i)\mid 1\leq i\leq q\}$.
\item The \textit{layer} $L^i_j\subseteq V(K_q^3)$ is a set of vertices fixing the $i$th coordinate as $j$. For example, the layer $L^1_j$ consists of pipes $P^{i}(j,b)$ for $i=2,3$ and $1\leq b\leq q$.
\item $C^i_j\subseteq L^i_j$ denotes the set of codewords in layer $L^i_j$, that is, for code $C\subseteq V(K_q^3)$ we have $C^i_j=C\cap L^i_j$.
\item $X^i_j\subseteq L^i_j$ denotes such non-codewords $v$ in $L^i_j$ that $I(C^i_j;v)=\emptyset$ and $X^i=\bigcup_{j=1}^q X^i_j$.
\item Let us denote $a^i_j=q-|C^i_j|$.
\item $M^i_j\subseteq L^i_j$ denotes the minimum dominating set of the induced subgraph $K^3_q[L^i_j]$ such that $C^i_j\subseteq M^i_j$. Note that $K^3_q[L^i_j]\simeq K_q\square K_q$ and hence, $|M^i_j|\geq q$.
\item Let us denote $f^i_j=|M^i_j|-q$. Note that $|X^i_j|\geq (f^i_j+a^i_j)^2$ and $f^i_j+a^i_j\geq0$ since $f^i_j=|M^i_j|-q\geq |C^i_j|-q=-a^i_j$, (\cite[pp. 15--16]{CubeCon}).
\end{itemize}

\begin{lemma}\label{f^i_j koko}
Let $C\subseteq V(K^3_q)$ and let $K_t\square K_t$ be a subgraph of $K^3_q[C^i_j]$ for some $i,j$. Then we have $f^i_j\geq t^2-t$.
\begin{proof}
We have $C^i_j\subseteq M^i_j$. Besides the vertices of $C^i_j$ inducing graph $K_t\square K_t$, there are $(q-t)^2$ vertices which are not dominated by the vertices of $C^i_j$. Moreover, we require at least $q-t$ vertices to dominate them. Hence, we have $|M^i_j|\geq t^2+(q-t)$ and thus, $f^i_j\geq t^2-t$.
\end{proof}
\end{lemma}

\begin{lemma}[\cite{CubeCon}, Lemma 10]\label{I-set perusom}

\noindent
Let $C$ be a code in $K^3_q$ and $v$ be a vertex of $K^3_q$.

\begin{itemize}
\item If a vertex $v$ has two codewords in its $I$-set and they do not locate within a single pipe, then there is exactly one other vertex which has those two codewords in its $I$-set.
\item  The $I$-set $I(v)$ is not a subset of any other $I$-set if and only if there are at least three codewords in $I(v)$ and they do not locate within a single pipe.
\end{itemize}
\end{lemma}

\begin{theorem}\label{DLD kuutio}
We have for $q\geq2$ $$\gamma^{DLD}(K^3_q)=q^2.$$
\begin{proof}
We have shown in \cite{CubeCon} that $\gamma^{SLD}(K^3_q)=q^2$. Hence, we have $\GSLD(K^3_q)\leq q^2$ by Corollary \ref{CorollarySLdtoDLD}. Let us assume that $C$ is an optimal solid-locating-dominating code in $V(K^3_q)$ with $|C|<q^2$. Since $|C|<q^2$, we have a layer, say $L^3_1$, with at most $q-1$ codewords and hence, we have $|X^3_1|\geq1$ by Lemma \ref{dominationNumber}. Let us assume that $(1,1,1)\in X^3_1$. Now, we have $(i,1,1)\not\in C$ for any $i$ and the same is true for $(1,j,1)$ for any $j$. Moreover, if we have $(1,1,h)\not\in C$, then $I(1,1,1)\subseteq I(1,1,h)$, a contradiction. Therefore, for each non-codeword in $X^i_j$ we have a pipe with $q-1$ codewords. Let us denote a pipe with $q-1$ codewords as $P^i_C(a,b)$ where $i$ denotes the direction of the pipe and $(a,b)$ denotes the coordinates in which the pipe intersects with the layer. Note that if $(a,b,z)\in X^3_z$ and $(a,b,z')\in X^3_{z'}$, then $z=z'$.

Let us first note that we have \begin{equation}\label{q+2 putkea}
|\{P_C^i(a,b)\mid 1\leq a,b\leq q\}|\leq q+1
\end{equation} for any fixed $i\in \{1,2,3\}$. Otherwise, we would have $|C|\geq (q+2)(q-1)=q^2+q-2>q^2-1$. Let us then consider the case where we have only $q-t$, $t\geq2$, codewords in a layer, say $L^3_1$. Then we have $|X^3_1|\geq t^2$ and these vertices (or some subset of them) induce subgraph $K_t\square K_t$ on $K^3_q$. Therefore, we have at least $t^2$ copies of codeword pipes $P^3_C(a,b)$ and without loss of generality, we may assume that values $(a,b)$ form the set $\{(i,j)\mid 1\leq i,j\leq t\}$. Thus, some subset of the vertices in $C^3_j$, for any fixed $j$ such that $2\leq j\leq q$, form an induced subgraph $K_t\square K_t$. Therefore, we have $f^3_j\geq t^2-t$ for any $2\leq j \leq q$ by Lemma \ref{f^i_j koko}. Thus, we have $$|X^3|\geq t^2+\sum_{j=2}^q (f^3_j+a^3_j)^2\geq t^2+\sum_{j=2}^q f^3_j+\sum_{j=2}^q a^3_j\geq t^2-t+(q-1)(t^2-t)+1=q(t^2-t)+1\geq 2q+1.$$ Note that $\sum_{j=2}^q a^3_j\geq 1-t$ and if $(a,b,j)\in X^3_{j}$, then $(a,b,i)\not\in X^3_i$ for each $i\neq j$. However, this is a contradiction with (\ref{q+2 putkea}). Therefore, we have $|C^i_j|\geq q-1$ for any $i,j$.

Let us then consider the case where $|C^3_1|=q-1$ and $C^3_1$ induces a discrete graph. Then for any non-codeword $v=(a,b,1)$, we have $|N(v)\cap C^3_1|\leq 2$ and the codewords in $N(v)\cap C^3_1$ do not locate within the same pipe. Therefore, by Lemma \ref{I-set perusom}, we have another non-codeword $w\in L^3_1$ such that $N(v)\cap C^3_1\subseteq N(w)$. Furthermore, this means that there is a codeword in $P^3(a,b)$. Since this is true for any non-codeword and $|L^3_1|=q^2$, we have $|C|\geq q^2$, a contradiction.

Let us then consider the case $|C^3_1|=q-1$ for $q\geq3$ and assume that some codewords in $C^3_1$ are neighbours. We may assume that $(1,1,1),(1,2,1)\in C^3_1$. Moreover, we may assume that $(q,q,1)\in X^3_1$. Since there are at least two codewords in the pipe $P^2(1,1)$ and there are $q-1$ codewords in $C^3_1$, we have at least two pipes $P^2(a,1)$ and $P^2(q,1)$ such that they contain no codewords. Therefore, we have $(a,q,1)\in X^1_1$. Moreover, we have codeword pipes $P_C^3(q,q)$ and $P_C^3(a,q)$. Now, we can consider layers $L^1_q$ and $L^1_a$. Let us first consider the layer $L^1_q$. First of all, it contains the codeword pipe $P_C^3(q,q)$ and since the pipe $P^2(q,1)$ contains no codewords, there has to be at least one codeword in every pipe $P^3(q,i)$ where $1\leq i\leq q-1$. Indeed, otherwise we would have $q-1$ codewords in some pipe $P_C^1(i,q)$, $2\leq i\leq q$, a contradiction with pipes $P^2(a,1)$ and $P^2(q,1)$ containing no codewords. Therefore, we have $|C^1_q|\geq 2q-2$. Furthermore, we get similarly $|C^1_a|\geq 2q-2$. However, now we have $$|C|\geq 2(2q-2)+\sum_{i=1,i\neq a}^{q-1}|C^1_i|\geq 2(2q-2)+(q-2)(q-1)=q^2+q-2>q^2-1,$$ a contradiction.\end{proof}
\end{theorem}

\section{Acknowledgement}
The authors would like to thank Mar\'ia Luz Puertas Gonzalez for fruitful discussions on
the topic.

\bibliographystyle{fundam}
\bibliography{LD}

\begin{thebibliography}{10}
\providecommand{\url}[1]{\texttt{#1}}
\providecommand{\urlprefix}{URL }
\expandafter\ifx\csname urlstyle\endcsname\relax
  \providecommand{\doi}[1]{doi:\discretionary{}{}{}#1}\else
  \providecommand{\doi}{doi:\discretionary{}{}{}\begingroup
  \urlstyle{rm}\Url}\fi
\providecommand{\eprint}[2][]{\url{#2}}

\bibitem{charon2002identifying}
Charon I, Hudry O, Lobstein A.
\newblock Identifying codes with small radius in some infinite regular graphs.
\newblock \emph{Electron. J. Comb.}, 2002.
\newblock \textbf{9}(1):R11.

\bibitem{cohen2001codes}
Cohen GD, Honkala I, Lobstein A.
\newblock On codes identifying vertices in the two-dimensional square lattice
  with diagonals.
\newblock \emph{IEEE on Trans. Comput.}, 2001.
\newblock \textbf{50}(2):174--176.

\bibitem{Trachtenberg}
Fazlollahi N, Starobinski D, Trachtenberg A.
\newblock Connected identifying codes.
\newblock \emph{IEEE Trans. Inform. Theory}, 2012.
\newblock \textbf{58}(7):4814--4824.

\bibitem{GWIDcpg}
Goddard W, Wash K.
\newblock {ID} codes in {C}artesian products of cliques.
\newblock \emph{J. Combin. Math. Combin. Comput}, 2013.
\newblock \textbf{85}:97--106.

\bibitem{DomInGraphs}
Haynes TW, Hedetniemi ST, Slater PJ.
\newblock Fundamentals of Domination in Graphs. 1998.
\newblock Marcel Dekker, New York, 1998.

\bibitem{complementLD}
Hernando C, Mora M, Pelayo IM.
\newblock Nordhaus--{G}addum bounds for locating domination.
\newblock \emph{European J. Combin.}, 2014.
\newblock \textbf{36}:1--6.

\bibitem{honkala2006optimal}
Honkala I.
\newblock An optimal locating-dominating set in the infinite triangular grid.
\newblock \emph{Discrete Math.}, 2006.
\newblock \textbf{306}(21):2670--2681.

\bibitem{honkala2006locating}
Honkala I, Laihonen T.
\newblock On locating--dominating sets in infinite grids.
\newblock \emph{European J. Combin.}, 2006.
\newblock \textbf{27}(2):218--227.

\bibitem{lowww}
Jean D, Lobstein A.
\newblock Watching systems, identifying, locating-dominating and discriminating
  codes in graphs: a bibliography.
\newblock \emph{Published electronically at }
  \url{https://dragazo.github.io/bibdom/main.pdf}.

\bibitem{JLLrntcld}
Junnila V, Laihonen T, Lehtil{\"a} T.
\newblock On regular and new types of codes for location-domination.
\newblock \emph{Discrete Appl. Math.}, 2018.
\newblock \textbf{247}:225--241.

\bibitem{CubeCon}
Junnila V, Laihonen T, Lehtil{\"a} T.
\newblock On a Conjecture Regarding Identification in Hamming Graphs.
\newblock \emph{Electron. J. Combin.}, 2019.
\newblock \textbf{P2}.

\bibitem{JLLnrclg}
Junnila V, Laihonen T, Lehtil{\"a} T.
\newblock New results on codes for location in graphs.
\newblock In: Proceedings of Russian Finnish Symposium on Discrete Mathematics
  2019. RuFiDiM 2019. 2019 pp. 105--116.

\bibitem{SLDDLDgraafit}
Junnila V, Laihonen T, Lehtil{\"a} T, Puertas ML.
\newblock On Stronger Types of Locating-dominating Codes.
\newblock \emph{Discrete Math. Theor. Comput. Sci.}, 2019.
\newblock \textbf{21}.

\bibitem{kcl}
Karpovsky MG, Chakrabarty K, Levitin LB.
\newblock On a new class of codes for identifying vertices in graphs.
\newblock \emph{IEEE Trans. Inform. Theory}, 1998.
\newblock \textbf{44}(2):599--611.

\bibitem{LT:disj}
Laifenfeld M, Trachtenberg A.
\newblock Disjoint identifying-codes for arbitrary graphs.
\newblock In: Proceedings International Symposium on Information Theory, 2005.
  ISIT 2005. IEEE, 2005 pp. 244--248.

\bibitem{RS:LDnumber}
Rall DF, Slater PJ.
\newblock On location-domination numbers for certain classes of graphs.
\newblock \emph{Congressus Numerantium}, 1984.
\newblock \textbf{45}:97--106.

\bibitem{WashThesis}
Rall DF, Wash K.
\newblock Identifying codes of the direct product of two cliques.
\newblock \emph{European J. Combin.}, 2014.
\newblock \textbf{36}:159--171.

\bibitem{Ray}
Ray S, Starobinski D, Trachtenberg A, Ungrangsi R.
\newblock Robust location detection with sensor networks.
\newblock \emph{IEEE J. Sel. Areas Commun.}, 2004.
\newblock \textbf{22}(6):1016--1025.

\bibitem{sampaio2024density}
Sampaio RM, Sobral GA, Wakabayashi Y.
\newblock Density of identifying codes of hexagonal grids with finite number of
  rows.
\newblock \emph{RAIRO Oper. Res.}, 2024.
\newblock \textbf{58}(2):1633--1651.

\bibitem{S:DomLocAcyclic}
Slater PJ.
\newblock Domination and location in acyclic graphs.
\newblock \emph{Networks}, 1987.
\newblock \textbf{17}(1):55--64.

\bibitem{S:DomandRef}
Slater PJ.
\newblock Dominating and reference sets in a graph.
\newblock \emph{J. Math. Phys. Sci.}, 1988.
\newblock \textbf{22}(4):445--455.

\bibitem{S:fault-tolerant}
Slater PJ.
\newblock Fault-tolerant locating-dominating sets.
\newblock \emph{Discrete Math.}, 2002.
\newblock \textbf{249}(1-3):179--189.

\end{thebibliography}


\end{document}